\documentclass[]{amsart}
\usepackage{latexsym,xy}
\xyoption{all}
\usepackage{graphicx,amscd,amssymb,latexsym,subfigure,verbatim,hyperref,epsfig,supertabular}

\newtheorem{defn0}{Definition}[section]
\newtheorem{prop0}[defn0]{Proposition}
\newtheorem{exp0}[defn0]{Experimental Results}
\newtheorem{disc0}[defn0]{Discussion}
\newtheorem{conj0}[defn0]{Conjecture}
\newtheorem{thm0}[defn0]{Theorem}

\newtheorem{corollary0}[defn0]{Corollary}
\newtheorem{example0}[defn0]{Example}
\newtheorem{remark0}[defn0]{Remark}
\newtheorem{que0}[defn0]{Question}

\newenvironment{defn}{\begin{defn0}}{\end{defn0}}
\newenvironment{prop}{\begin{prop0}}{\end{prop0}}
\newenvironment{conj}{\begin{conj0}}{\end{conj0}}
\newenvironment{thm}{\begin{thm0}}{\end{thm0}}
\newenvironment{expr}{\begin{exp0}}{\end{exp0}}
\newenvironment{disc}{\begin{disc0}}{\end{disc0}}

\newenvironment{cor}{\begin{corollary0}}{\end{corollary0}}
\newenvironment{exm}{\begin{example0}\rm}{\end{example0}}
\newenvironment{remark}{\begin{remark0}\rm}{\end{remark0}}
\newcommand{\F}{{\mathbb F}}

\def\ZZ{\mathbb{Z}}
\def\CC{\mathbb{C}}
\def\QQ{\mathbb{Q}}
  
\def\P{\mathbb{P}}

\numberwithin{equation}{section}

\begin{document}
\title[Codes from surfaces with small Picard number]%
{Codes from surfaces with small Picard number} 

\author{John Little}
\address{Department of Mathematics and Computer Science\\
College of the Holy Cross, Worcester, MA 01610}
\email{jlittle@holycross.edu}

\author{Hal Schenck}
\thanks{Schenck supported by NSF 1312071}
\address{Department of Mathematics, Iowa State University, 
Ames, IA 50011}
\email{hschenck@iastate.edu}


\begin{abstract}
Extending work of M. Zarzar, we evaluate the potential of Goppa-type 
evaluation codes constructed from linear systems on projective algebraic surfaces with 
small Picard number.  Putting this condition on the Picard number provides some control over the 
numbers of irreducible components of curves on the surface and hence over the 
minimum distance of the codes.  We find that such surfaces do not automatically produce good 
codes; the sectional genus of the surface also has a major influence.  Using that additional invariant, 
we derive bounds on the minimum distance under the assumption that the hyperplane section class 
generates the N\'eron-Severi group.   We also give several examples of codes from such surfaces with 
minimum distance better than the best known bounds in Grassl's tables.
\end{abstract}

\maketitle

\section{Introduction}
A convenient reference for all of the coding theory concepts we will need is \cite{hp}.
We first recall a general setup for defining error-correcting
codes from algebraic varieties $X$ defined over $\F_q$.  The \emph{evaluation, 
or functional code} produced from a subset ${\mathcal S} = \{P_1,\ldots,P_n\}$ of 
the $\F_q$-rational points on $X$ and an $\F_q$-vector space of functions ${\mathcal F}$ 
defined on ${\mathcal S}$ is defined as follows.  The set of codewords is the image of the mapping:
\begin{eqnarray}
\label{ev}
ev_{\mathcal S} : {\mathcal F} & \longrightarrow & \F_q^n\\
                f &\mapsto & (f(P_1),\ldots,f(P_n)).\nonumber
\end{eqnarray}
When ${\mathcal F}$ is the vector space of  
sections of a line bundle ${\mathcal O}_X(D)$ for some divisor $D$ on $X$
defined over $\F_q$, we will call the resulting evaluation code 
$C({\mathcal S},D,\F_q)$.  Several general statements about the
parameters $[n,k,d]$ of $C({\mathcal S},D,\F_q)$ follow from these definitions:
\begin{itemize}
\item $n$ is $|{\mathcal S}|$, 
\item $k$ is equal to the dimension of the space of functions
obtained by restricting the elements of ${\mathcal F}$ to ${\mathcal S}$.  
(This equals $\dim {\mathcal F}$ unless there are nonzero elements of ${\mathcal F}$ 
vanishing on all of ${\mathcal S}$.) 
\item $d$ is determined by the maximal number of zeroes of a section of ${\mathcal O}_X(D)$ 
at the points in ${\mathcal S}$.
\end{itemize}
Hence the properties of these codes are closely tied to the algebraic
geometry of the variety $X$ and the divisor $D$.  The well-known Goppa codes from 
algebraic curves are first examples. The survey \cite{l} includes
a discussion of previous work on codes obtained from various types of 
varieties by similar methods.  

In \cite{ls}, in the case that $X$ is a toric surface, under the hypothesis that
 $q$ is sufficiently large, we obtained upper and lower
bounds on the minimum distance of toric surface codes by relating codewords of 
small weight to reducible sections of the line bundle $\mathcal{O}_X(D)$ and 
Minkowski summands of the polygon $P_D$ defined by $D$.  Our results were 
later extended and improved by Soprunov and Soprunova in \cite{ss}.  

In \cite{z}, M. Zarzar suggested that surfaces $X$ in $\P^3$ with small
Picard number over $\F_q$ might lead to good evaluation codes
and provided a number of examples in support of this idea.  Zarzar's approach in 
\cite{z} is related to, but is in some sense orthogonal to that of \cite{ls} and 
\cite{ss}; the idea is to impose conditions so that not too much factorization of 
sections can occur.  The methods employed in \cite{z} are specific to the case of 
surfaces in $\P^3$, and Voloch-Zarzar give an extension to surfaces in $\P^n$ in \cite{vz}.  
Couvreur used these ideas to produce some very good codes in \cite{c}.

Our goals in this paper are to refine and extend the results of \cite{z} and \cite{vz}, 
to report the results of some experimentation with this approach, 
and to indicate both some situations where this approach succeeds in producing good codes
and some where it does not.  

We will begin with some general observations about this class of 
codes in \S 2.  The definition of the N\'eron-Severi group and 
its role in bounding the minimum distance will be 
reviewed in \S 3.  
In \S 4, we focus on anticanonical rational surfaces; specifically
surfaces obtained by blowing up a small number of points in $\P^2$. 
Classical work of Swinnerton-Dyer \cite{sd} plays a key role here.  We close with
an example (in effect, a variation of a construction studied in \cite{c}) 
where our bounds are better than previous bounds from \cite{b} and \cite{d}.
We will obtain some codes better than the best 
known examples in Grassl's tables \cite{g}. 

\section{General coding-theoretic observations}

All of the examples we consider will come from projective surfaces $X \subset \P^r$, 
$r\ge 3$, given by explicit sets of homogeneous polynomial equations.  In this situation, 
when $D = sH$ is a positive integer multiple of a hyperplane
section of the surface, it is possible to work with a more concrete version of the evaluation codes
introduced above.  Namely, we can take ${\mathcal S} = X(\F_q)$, 
the whole set of $\F_q$-rational points of $X$.  We then choose a homogeneous
coordinate vector for each point normalized so that the rightmost nonzero entry is $1$: 
that is there exists some $\ell$, $0 \le \ell \le r$, such that
$$
P = (x_0 : \cdots : x_\ell = 1 : 0 : \cdots : 0).
$$
We write $\F_q[x_0,\ldots,x_r]_s$ for the vector space of homogeneous polynomials
of degree $s$.  Each  $f \in \F_q[x_0,\ldots,x_r]_s$ has a well-defined 
value in $\F_q$ at the $P$ as above.  

\begin{defn}  
\label{evcode}
The evaluation code from ${\mathcal S} = X(\F_q)$ in this form and 
${\mathcal F} = \F_q[x_0,\ldots,x_r]_s$ will be denoted $C(X,s,\F_q)$.  
\end{defn}

\begin{remark}
\label{projnormal}
The code defined this way will coincide with the algebraic geometry
code constructed from $H^0(X,{\mathcal O}(s))$ when $X$ is projectively
normal; if that condition is not satisfied, the $C(X,s,\F_q)$ code may
only give a subcode of the algebraic geometry code.
\end{remark}

As in \S 1.7 of \cite{hp}, we will say that a linear code $C$ with generator matrix $G$
is \emph{monomially equivalent} to a second code $C'$ if there exists an
$n\times n$ monomial matrix $M$ such that $G' = GM$ is a generator matrix for $C'$.  
A monomial matrix has the form $M = DP$ where $D$ is an invertible $n\times n$ 
diagonal matrix and $P$ is a permutation matrix.  Monomial equivalence is a natural
equivalence relation to use to classify our codes because it preserves all the quantities
of interest -- $n, k, d$ and the number of codewords of each weight (the weight enumerator).

If $X$ is projectively normal (see Remark~\ref{projnormal} above), and $f \in k[x_0,\ldots,x_r]_s$ 
with $s \ge 1$ is fixed, then the global sections of ${\mathcal O}_X(sH)$ can be
identified with the rational functions $g/f$ for $g \in k[x_0,\ldots,x_r]_s$.  If 
there exists such an $f$ that vanishes at none of the points in $X(\F_q)$, then it is easy 
to see that the code $C(X(\F_q),sH,\F_q)$ described in the introduction is monomially 
equivalent to this one.  

We will need a way to say, roughly speaking, that the $\F_q$-rational points on a variety 
over $\F_q$ determine the equations of $X$ in all sufficiently low degrees.  For 
instance if $X$ is a hypersurface of degree $m$, we will mean that there are no forms
of degree $< m$ vanishing on $X(\F_q)$ and in degree $m$ the only forms vanishing on 
$X(\F_q)$ also vanish on $X$.  We will require similar conditions in higher codimensions.

\begin{defn}
\label{Fqgen}
A projective variety $X$ over $\F_q$ is $\F_q$-\emph{general} if the homogeneous ideal of $X$
is generated in degrees $\le m$ for some $m \ge 1$ and for all $\ell$ with $1 \le \ell \le m$, 
every homogenous form of degree $\ell$ that vanishes on $X(\F_q)$ also vanishes on $X$.  
\end{defn}

\begin{exm}
\label{nongen}
An example of a variety that does not satisfy this definition is  
a conic $X$ in $\P^2$ over $\F_q$ whose defining equation of degree $m = 2$
factors into two Frobenius-conjugate linear forms with coefficients in $\F_{q^2}$.  
Then $X(\F_q)$ consists of just the point of intersection of the two conjugate 
$\F_{q^2}$-lines, and the condition in Definition~\ref{Fqgen} fails for $\ell = 1$ and $\ell = 2$.  
\end{exm}

Varieties that are not $\F_q$-general are essentially of no use in our construction, 
so we will assume from now on that this $\F_q$-generality condition holds. 
  
The code $C(X,1,\F_q)$ plays a special role.  Note that if $X \subset \P^r$, then
$k = r + 1$ and using the standard basis $\{x_0,x_1,\ldots,x_r\}$ for $\F_q[x_0,\ldots,x_r]_1$,
the columns of the standard generator matrix will consist of the normalized homogeneous 
coordinate vectors of points $P \in X(\F_q)$.  This observation has some immediate consequences.

\begin{prop}
\label{monequiv}
Let $X$ and $X'$ be varieties, both defined over $\F_q$.  
\begin{itemize}
\item[(i)]  If $X$ and $X'$ are projectively equivalent over $\F_q$, then the code $C(X,1,\F_q)$ is 
monomially equivalent to $C(X',1,\F_q)$.
\item[(ii)]  Assume also that both $X$ and $X'$ are $\F_q$-general.  
Then the converse also holds: if $C(X,1,\F_q)$ is monomially equivalent to 
$C(X',1,\F_q)$ then $X$ and $X'$ are projectively equivalent over $\F_q$.
\end{itemize}
\end{prop}

\begin{proof}
(i)  If $X$ and $X'$ are projectively equivalent over $\F_q$, there is an invertible $(r + 1) \times (r + 1)$
matrix $A$ over $\F_q$ defining a projective linear transformation $\phi_A(P) = AP$ on $\P^r$ such that 
$\phi_A(X) = X'$.  By our hypotheses $\phi_A$ also maps $X(\F_q)$ bijectively to $X'(\F_q)$.  
It follows that, after adjusting with an invertible diagonal matrix $D$ to normalize the columns 
(the homogeneous coordinate vectors of the points in $X'(\F_q)$) to the form described above, 
the matrix $G' = (AG)D$ will be the standard generator matrix
for $C(X',1,\F_q)$.  But $AG$ is also a generator matrix for $C(X,1,\F_q)$ since $A$ is invertible.  
Hence the two codes are monomially equivalent.  

(ii) Conversely, suppose $C(X,1,\F_q)$ and $C(X',1,\F_q)$ are monomially equivalent.  If we take the 
standard generator matrix $G$ of the first code, this means that there exists an invertible diagonal
matrix $D$ and a permutation matrix $P$ such that $G(DP)$ is a generator matrix for $C(X',1,\F_q)$.
The rows of this matrix correspond to some basis for the vector space of linear forms with coefficients
in $\F_q$.  So there is an invertible $(r + 1) \times (r + 1)$
matrix $A$ over $\F_q$ such that $A(G(DP)) = (AG)(DP)$ is a generator matrix
for $C(X',1,\F_q)$ whose rows correspond to the basis $\{x_0,x_1,\ldots,x_r\}$.  Multiplying
by another invertible diagonal matrix $D'$ on the right to normalize the columns if necessary, we will obtain
the standard generator matrix of $C(X',1,\F_q)$ in the form $(AG)(DPD')$. 
By an easy computation with matrices of these
forms, the product $DPD'$ can be rewritten in the form
$D''P$ for another invertible diagonal matrix $D''$.  So we have $G' = (AG)(D''P) = (AGD'')P$ and
$G'$ and $AGD''$ differ only by a permutation of the columns.  So the columns of $AGD''$ are just 
the normalized forms of the images of the points in $X(\F_q)$ under $\phi_A$. This shows
$\phi_A(X(\F_q)) = X'(\F_q)$, and we claim that $\phi_A(X) = X'$ because of the assumption of $\F_q$-generality.
This follows since if $f$ is any homogeneous polynomial in the vanishing ideal of $X$, then $f \circ \phi_A^{-1}$
vanishes on $X'(\F_q)$ and hence on all of $X'$.  
Therefore $X$ and $X'$ are projectively equivalent over $\F_q$.
\end{proof}

\section{Bounds in terms of the Picard number}

From now on, we will consider only codes constructed from smooth, projective, absolutely
irreducible surfaces $X$ defined over $\F_q$.  In addition, we will fix a particular embedding 
$X \hookrightarrow \P^r$ and work with particular polynomial equations for $X$.  We will not repeat all of those
hypotheses each time we refer to a surface $X$.  The \emph{N\'eron-Severi group} 
of $X$ over $\F_q$, denoted by ${\rm NS}(X) = {\rm NS}_{\F_q}(X)$, is the group of $\F_q$-rational divisors on $X$ 
modulo algebraic equivalence.   This is a finitely generated abelian group with a free part
isomorphic to $\ZZ^\rho$ for some $\rho \ge 1$.  The
rank $\rho=\rho(X)$ is known as the \emph{Picard number} of $X$.   We note that this is sometimes
called the \emph{arithmetic Picard number} to distinguish it from the {\it geometric Picard number} of
the corresponding variety $X$ over the algebraic closure of $\F_q$.  In this article, the phrase 
``Picard number'' with no modifier will always refer to the arithmetic Picard number.
The relevance of the Picard number for 
coding theory is based on the following facts from \cite{vz} and \cite{z}.  

\begin{prop}[\cite{vz}, Lemma 2.2] 
\label{VZ}
If ${\rm NS}(X)$ is generated by the class of an
ample divisor $H$ and $[D] = [mH]$, then the divisor of zeroes of a nonzero element of $L(D)$
has at most $m$ distinct $\F_q$-irreducible components.   
\end{prop}

\begin{proof}  For the convenience of the reader, we reproduce the proof in \cite{vz}.
Suppose $f \in L(D)$ is a nonzero rational function, and consider the 
decomposition
$$
(f)_0 = a_1 D_1 + \cdots + a_\ell D_\ell,
$$
where $D_i$ are $\F_q$-irreducible and $a_i$ are integers $\ge 1$.  
Then since $D_i$ is algebraically equivalent to
$b_i H$ for some integer $b_i \ge 1$ over $\F_q$,
$$\ell H^2 \le \sum_i a_i b_i H^2 = (f)_0 H = (f)_\infty H \le m H^2.$$
This implies $\ell \le m$.   
\end{proof}  

For surfaces in $\P^3$, \cite{z} provides a somewhat sharper statement.

\begin{prop}[\cite{z}, Lemma 2.1]  Let $X$ be
a surface of degree $a$ in $\P^3$, where $a$
is not divisible by the characteristic of $\F_q$.  If 
${\rm NS}(X)$  is generated by the class of an effective divisor $D$ 
(so the Picard number of $X$ is 1) and $Y$ is an absolutely irreducible surface of degree $1 \le m < a$, 
then $X \cap Y$ is absolutely irreducible.
\end{prop}

In this case, ${\rm NS}(X)$ is generated by the class of a plane section, and the 
statement follows.  Applying Proposition~\ref{VZ} and the Hasse-Weil-Serre bound we have the 
following bound on $n - d_1$, where $d_1 = d(C(X,1,\F_q))$.  

\begin{thm} 
\label{d1bdth}
\label{D1Th}
Let $H$ be an absolutely irreducible $\F_q$-hyperplane section of a surface $X$.   Assume the class of
$H$ generates ${\rm NS}(X)$ over $\F_q$.  Then
\begin{equation}
\label{d1bd}
n - d_1 \le 1 + q + \pi \lfloor 2\sqrt{q}\rfloor,
\end{equation}
where $\pi$ is the sectional genus of $X$ (the arithmetic genus of the curves in the hyperplane section class). 
\end{thm} 

\begin{proof}
If $f$ is any linear form defined over $\F_q$, then the divisor $D = X \cap V(f)$ is linearly equivalent, 
hence algebraically equivalent, to $H$ on $X$.   Taking $m = 1$ in 
Proposition~\ref{VZ},  $D$ must also be irreducible over $\F_q$.
Any codeword produced by evaluation of a linear form $f$ where $D$ is also absolutely
irreducible will have at most $1 + q + \pi \lfloor 2\sqrt{q}\rfloor$ zero entries
by the form of the Hasse-Weil-Serre bound for singular curves as in \cite{ap}.   

It remains to consider the case that 
$D$ is irreducible over $\F_q$ but not absolutely irreducible.  In this case by 
Theorem 3.2 of \cite{lr}, there is some $r \ge 2$ such that $D = D_1 \cup \cdots \cup D_r$
with all $D_i$ absolutely irreducible and defined over $\F_{q^r}$.   These absolutely
irreducible components are permuted by the action of the cyclic group
${\rm Gal}(\F_{q^r}/\F_q)$.  
By Lemma 2.8 and Proposition 3.8 of \cite{lr}, 
$$D(\F_q) = {\rm Sing}(D)(\F_q) = (D_1 \cap \cdots \cap D_r)(\F_q).$$
Hence for any pair $i\ne j$, 
$$|D(\F_q)| \le D_i \cdot D_j,$$
using the intersection form on $X$.  

By the \emph{adjunction formula} on $X$ (see, for 
instance \cite{H}, Proposition V.1.5), we have
\begin{equation}
\label{adj}
2\pi - 2 = (K_X + H)\cdot H,
\end{equation}
where $K_X$ is the canonical divisor class on $X$.  
On the other hand, the reducible divisor $D$ has the same arithmetic genus as the irreducible
curve $H$  since $D$ is linearly equivalent to $H$.  Since $D = D_1 +\cdots + D_r$ as divisors
on $X$, the bilinearity
of the intersection form gives
\begin{align}
\label{arithgen}
\nonumber
2\pi - 2 &= (K_X + D) \cdot D\\ 
&= \sum_{i = 1}^r (K_X + D_i)\cdot D_i  + 2\sum_{i < j} D_i \cdot D_j\\
&= \sum_{i = 1}^r (2p_a(D_i) - 2)  + 2\sum_{i < j} D_i \cdot D_j.  \nonumber
\end{align}
The $D_i$ are irreducible over the algebraic closure of $\F_q$ so their arithmetic genus
satisfies $p_a(D_i) \ge 0$.  Hence for all $i$, 
$$2 p_a(D_i) - 2 = (K_X + D_i)\cdot D_i \ge -2.$$
Hence rearranging the last equality from \eqref{arithgen}, we see for all pairs $i \ne j$, 
\begin{equation}
\label{exc}
D_i \cdot D_j \le \frac{1}{r(r - 1)} \left(2\pi - 2 + 2r\right)  = \frac{1}{r(r-1)}(2 \pi - 2) + \frac{2}{r-1}.
\end{equation}
Comparing the Hasse-Weil-Serre upper bound and the bound from \eqref{exc}, since $r \ge 2$ and 
$\frac{2}{r - 1} \le 2$, 
$$ \frac{1}{r(r-1)}(2 \pi - 2) <  \pi \lfloor 2\sqrt{q}\rfloor $$
and 
$$\frac{2}{r - 1} < 1 + q$$
for all $q$.  This shows that any curves in the class of $H$ that are irreducible over $\F_q$, but not absolutely
irreducible, have fewer $\F_q$-rational points than the 
Hasse-Weil-Serre upper bound for absolutely irreducible curves.  
\end{proof}

\begin{exm}
An example related to the proof of Theorem~\ref{D1Th} comes from the elliptic quadric surfaces $X$
in $\P^3$.  The Picard number is $\rho(X) = 1$ in this case and the N\'eron-Severi group is generated by the 
class of a smooth plane section (a plane conic).  There are also tangent planes intersecting
$X$ in pairs of lines, but these lines are only defined over the quadratic extension of 
$\F_q$.  Hence we have reducible divisors $D = D_1 + D_2$ in this class with $D_i$ lines
over the algebraic closure of $\F_q$.  The only $\F_q$-rational point on such a section is the point of tangency (where
the two lines meet).  The upper bound in \eqref{exc} is achieved in this case. This is the situation described in Example~\ref{nongen}.
\end{exm}

We also consider the codes $C(X,s,\F_q)$ for $s > 1$ and set $d_s = d(C(X,s,\F_q))$. 
Our main observation here is that, at least when $q$ is sufficiently large, the minimum distance 
of  $C(X,s,\F_q)$ is, to some extent, controlled by the minimum distance of $C(X,1,\F_q)$. 
The simplest statement here concerns $n - d_s$, the largest number
of zeroes in any nonzero codeword.   

\begin{thm}
\label{distbds}  Under the hypotheses of Theorem~\ref{D1Th}, if $q$ is sufficiently large, then
$$n - d_s \le s(n - d_1).$$
\end{thm}

\begin{proof}
Consider the divisor $D = X \cap V(f)$ for $f$ of degree $s$.
We begin with the case that 
$D = X \cap V(f)$ has a decomposition $D = D_1 \cup \cdots \cup D_r$
into absolutely irreducible components defined over 
$\F_q$.  By Corollary 3 of \cite{AP2}, 
$$|D(\F_q) - (r q + 1)| \le 2\pi_D \sqrt{q},$$
where $\pi_D$ is the arithmetic genus (which depends only on $s$).   
We note that the curve $D$ is absolutely connected because
$D$ is an ample divisor on $X$ (see Corollary III.7.9 of [H]).  
 Among the forms of degree $s$ are those that factor over 
$\F_q$ into products of linear forms.  
By arguments similar to those we used in \cite{ls},  
it is not difficult to see that when $q$ is sufficiently large,
 the lower bound 
$$s q + 1 - 2 \pi_D \sqrt{q}$$
on $|D(\F_q)|$ for sections with $s$
components is greater than the upper bound 
$$r q + 1 + 2 \pi_D \sqrt{q}$$
for sections with $r$ components for all $r < s$.  
The maximum possible number of points on a
reducible $D$ with a decomposition of this form is attained for 
some $D$ with $r = s$  and hence each $D_i$ is a hyperplane section.  Hence $D$
has at most $s(n - d_1)$ $\F_q$-rational points, and the corresponding
codewords have at most $s(n - d_1)$ zero entries.  

It remains to consider the case where $D = X \cap V(f)$ has some
$\F_q$-irreducible component that is not absolutely irreducible.  
But then the argument given in the proof of Theorem~\ref{D1Th}
shows that the largest possible number of $\F_q$-points on any such 
component can only decrease relative to the case where that $\F_q$-irreducible 
component is absolutely irreducible.

Hence when $q$ is large enough, 
we will have $n - d_s \le s(n - d_1)$.
\end{proof}

We will not address the problem of determining 
precise bounds on the $q$ for which the conclusion of Theorem~\ref{distbds} holds in this article because
this tends to depend on particular properties of the surface $X$ and the divisor $D$.
To conclude this section, we record an observation that may be of independent interest.  It 
is possible to use these properties of evaluation codes to deduce properties of the surfaces 
they come from.

\begin{cor} Let $X$ be a smooth projective surface defined over $\F_q$ and assume that
the bound given in \eqref{d1bd} does not hold.  Then the class of the hyperplane
section $H$ does not generate ${\rm NS}(X)$.
\end{cor}

\begin{proof}  If \eqref{d1bd} does not hold, then some hyperplane section of $X$
must be reducible and the irreducible components give elements of ${\rm NS}(X)$ that
are not in the subgroup generated by the class of $H$.  
\end{proof}

In the situation of the corollary, if the components come from different algebraic 
equivalence classes, then in particular $\rho(X) > 1$.

\section{Picard number 1 (or small) is not enough}

A first obstacle to making use of these ideas is simply the problem of finding 
explicit surfaces with small Picard number.  

\subsection{Codes from smooth cubic surfaces in $\P^3$}

The article \cite{z} uses smooth cubic surfaces in $\P^3$ as a test case.
In this subsection, we provide some more detailed information about these codes and show that
an example from \cite{z} is actually the best possible for codes from these surfaces in the 
situation considered there.

Over the algebraic closure of $\F_q$, a smooth cubic surface is obtained as the blow-up 
of $\P^2$ in six points in general position and always has 27 lines as in the classical situation 
over $\CC$.  Moreover, the orthogonal complement of the canonical class in the Picard
group can be identified with the $E_6$ lattice and the Weyl group of $E_6$ acts on the lines.
However the lines may only be defined over an extension of $\F_q$ and some
cubics contain no $\F_q$-rational lines.  The article \cite{sd} gives a 
classification of the possibilities according to the action of the Frobenius automorphism on the set 
of lines.  That classification is based on the conjugacy class of the Weyl group 
containing the Frobenius automorphism.  According to Table 1 of \cite{sd}, there are exactly 
five possible types of cubics with $\rho(X) = 1$; such an $X$ contains no $\F_q$-rational 
lines or conics.  The types are denoted
as follows:  

\begin{equation}
\label{eq:CubicsRhoOne}
\begin{array}{cccc}
      {\rm Class} & {\rm Permutation\ Type} & N_1 = |X(\F_q)| & {\rm ord}(\eta_j) \\ \hline
          C_{10} & \{3,6^3,6\} & q^2 - q + 1 & 2,2,3,3,6,6\\
          C_{11} & \{3^9\}     & q^2 - 2q + 1& 3,3,3,3,3,3\\
          C_{12} & \{3,6^4\}   & q^2 + 2q + 1& 3,3,6,6,6,6\\
          C_{13} & \{3,12^3\}  & q^2 + 1 & 3,3,12,12,12,12\\ 
          C_{14} & \{9^3\}     & q^2 + q + 1 &9,9,9,9,9,9\\
\end{array}
\end{equation}

The second column indicates how the 27 lines on the surface over the algebraic closure
are permuted by the Frobenius automorphism of $X$.   
To give some idea of the intricacies here, on $C_{10}$ and $C_{12}$ cubics there is one orbit of length 
3 consisting of three coplanar lines defined over $\F_{q^3}$; the plane containing them is hence defined over $\F_q$.
The other 24 lines on a $C_{10}$ or $C_{12}$ are permuted in 4 orbits of size 6.   But those orbits of length 6
have different structures.  According to a discussion on p. 59 of \cite{sd}, on 
a $C_{10}$ cubic, three of the orbits consist of three skew lines and their three
transversals, but none of the orbits on a $C_{12}$ has that form.  On the other hand, one
of the orbits on a $C_{12}$ consists of two triangles in planes interchanged by the square of the Frobenius.
(Those triangles split over $\F_{q^6}$.)  

In the last column, the $\eta_j$ give a description of the 6 nontrivial reciprocal roots of 
the factor in the denominator of the zeta function of $X$ corresponding to the $H^2$ cohomology 
group:  $\alpha_j = \eta_j q$,
where $\eta_j$ is a primitive ${\rm ord}(\eta_j)$th root of unity.  From 
the specific shape of the zeta function
for smooth cubics in $\P^3$ (in particular, the fact that both factors corresponding to the $H^1$ and $H^3$ 
cohomology are trivial),
knowing only the $\eta_j$ allows us to compute 
$N_r = |X(\F_{q^r})|$ for all $r \ge 1$.    

A recent paper by Rybakov and Trepalin, \cite{rt}, studies when cubic surfaces of these
types exist over particular finite fields.

\begin{expr}
With a random search, we found cubic surfaces of each of these types over $\F_7$ and constructed
the codes $C(X, k, \F_7)$ for $k = 1, 2$.  The following list gives parameters for the corresponding 
codes for $k=1$. The best possible 
$d$ values are taken from \cite{g}.

\begin{itemize}
\item $C_{10}$ -- $[43,4,30]$, $[43,4,31]$ (best possible $d$ for $n = 43, k = 4$ is $d=35$)
\item $C_{11}$ -- $[36,4,23]$, $[36,4,24]$ examples (best possible $d$ is $28 \le d \le 29$)   
\item $C_{12}$ -- $[64,4,51]$  (best possible $d$ is $52 \le d \le 53$)
\item $C_{13}$ (very rare) -- found  $[50,4,37]$  (best possible $d = 42$)
\item $C_{14}$ (rare) -- found $[57,4,44]$ (best possible $d = 47$)
\end{itemize}

The frequent occurrence of the number $n - d = 13$ here can be explained as follows.
Recall the Hasse-Weil-Serre bound:  If $X$ is a smooth curve of genus $g$ over $\F_q$, then 
$$|X(\F_q) - (1 + q)| \le g\cdot \lfloor 2\sqrt{q}\rfloor.$$
Hence an upper bound on the number of $\F_7$-points on 
a smooth plane cubic is $1 + 7 + \lfloor 2\sqrt{7}\rfloor = 13$ and this
bound is achieved, for instance, for the Weierstrass form cubic
$y^2z = x^3 + 3z^3$.  
Moreover, singular (but irreducible) plane sections all have either $q = 7$ (``split'' node), 
$q + 1 = 8$ (cusp), or $q + 2 = 9$ (``non-split'' node)
$\F_7$-points.  Note that some of the $C_{10}$ and $C_{11}$ surfaces have no plane sections 
with $13$ $\F_7$-points, though.  One more general observation is that if one delves into
the deeper structure of these codes, for instance by using {\tt Magma} to compute the weight 
enumerators for a large collection of codes from cubics, one striking feature is how 
\emph{variable} they are.  Even among codes with the same $d$, there will be many different 
nonequivalent codes.   This is a reflection of Proposition~\ref{monequiv} from the last section, 
of course.  The $C(X,1,\F_q)$ codes effectively encode all the structure of $X$ up to projective 
equivalence.
\end{expr}
  
As Zarzar observed, among the cubics with Picard number 1, the $C_{12}$ cubics are clearly 
the best for construction of codes.
This is also provides some confirmation of Zarzar's \emph{Ansatz} concerning the Picard number.  
Cubics with $\rho(X) > 1$ contain lines rational over $\F_q$ and hence
have reducible plane sections with 2 or more irreducible rational components.  The $C(X,1,\F_q)$
codes will have $n - d > 13$, though $n$ can also be as large as $q^2 + 7q + 1$.  

Evaluating quadrics, Zarzar also reported a code $C(X, 2, \F_7)$ from a $C_{12}$ cubic surface
with parameters $[64, 10, 38]$.  Here, following the pattern from Proposition~\ref{distbds}, 
the minimum weight codewords come from reducible quadrics 
intersecting the cubic $X$ in reducible curves with two components, both smooth plane cubics 
with $13$ $\F_7$-rational points.  Moreover, the line of intersection of the two planes meets the cubic 
in three points rational over $\F_{q^3}$, so the resulting codewords have exactly $13 + 13 = 26$ 
zeroes.  By contrast, smooth intersections of a cubic and a quadric in $\P^3$ are curves of 
genus 4 with at most 24 $\F_7$-rational points, according to \cite{mp}.  The Hasse-Weil-Serre upper
bound $1 + 7 + 4\cdot \lfloor 2\sqrt{7}\rfloor = 28$ is not achieved in this case.
  
A natural question to ask here is how well the codes from $C_{12}$ cubics 
might do over other, larger fields.  Based on a large amount of experimental data, 
we offer the following conjecture that bears on this question.  
We say an irreducible plane cubic curve is {\it optimal} over $\F_q$ if it contains
the largest possible number of $\F_q$-rational points over all plane cubics.  As above
with $g = 4$, the Hasse-Weil-Serre bound is not always achieved when $g = 1$.  But 
compared to the higher genus case, much more is known about when the bound is not reached.
In particular, when $g = 1$ and the bound is not reached, there exist curves with
$q + \lfloor 2\sqrt{q}\rfloor$ $\F_q$-rational points, that is, just one less than the bound.
See \cite{aper}.  

\begin{conj}
\label{C12Conj}
For all $q \ge 7$, $C_{12}$ cubic surfaces over $\F_q$ always contain 
optimal cubic plane sections.
\end{conj}

It is possible to normalize the equations of $C_{12}$ cubics using the description of the 
lines on these surfaces given in the paragraph after the table in Equation~\eqref{eq:CubicsRhoOne} above.
Up to projective equivalence in $\P^3$, 
we can put the triangle orbit in the plane $w = 0$, so the plane section 
$X \cap V(w)$ is defined by an equation of the form $L \cdot F(L) \cdot F^2(L)$
where $L$ is a linear form in $x,y,z$ with coefficients
in $\F_{q^3}$.  Letting  $M$ be a second linear form in $x,y,z,w$ with coefficients in $\F_{q^2}$, then 
any $C_{12}$ cubic is projectively equivalent to one given by an equation of the form
\begin{equation}
\label{eq:CayleySalmon}
L\cdot F(L) \cdot F^2(L) = w\cdot M \cdot F(M).
\end{equation}
This is a special case of the \emph{Cayley-Salmon form} of the equation of a cubic
surface studied in classical algebraic geometry; see \cite{do}.
The Frobenius orbit of the line defined by $L = M = 0$ consists of the 6 lines defined by 
\begin{align*} 
L = M = 0, &\quad
F(L) = F(M) = 0\\
F^2(L) = M = 0, &\quad
L = F(M) = 0\\
F(L) = M =0, &\quad
F^2(L) = F(M) = 0.
\end{align*}
The first, third and fifth of these lie in the plane $M = 0$ and the other three lie in $F(M) = 0$, 
and those two triangles are interchanged by $F$.  

\begin{expr}
By considering equations of the form \eqref{eq:CayleySalmon}, it is feasible, when $q$ is small,
to generate surfaces in every projective equivalence class over $\F_q$ by varying the coefficients in the 
linear forms $L$ and $M$, then test all of them for existence of optimal hyperplane sections.  Using this approach, we 
have verified the conjecture for $q \le 9$.  
We have also done extensive random searches for all $q \le 37$ that have produced no 
counterexamples.  We hope to return to this conjecture in future work.  

Over $\F_7$, our results also show that there are sufficiently many such optimal plane 
sections that there always exist two in distinct planes whose line of intersection meets 
the cubic in no points rational over $\F_7$.  Hence the $C(X,2,\F_7)$ codes from $C_{12}$ cubics 
{\it always} have the parameters $[64,10,38]$ found by Zarzar.  
\end{expr}

Since the best known code from \cite{g} with
$n = 64, k = 10$ over $\F_7$ has $d = 41$, we have another 
strong heuristic indication that {\it in addition to limiting the Picard number, the genus of the 
components of maximally reducible curves on $X$ in the linear system giving the evaluation
functions must also be controlled if the construction of this paper is to produce really good codes}.   
Similarly, if the conjecture above holds in general, we expect that (using the notation from
Theorem~\ref{distbds}) the $C(X,2,\F_q)$ codes will always have $n - d_2 = 2(n-d_1)$, 
i.e. twice the optimal number of $\F_q$-rational points on a curve of genus 1.  

\subsection{Higher degree surfaces in $\P^3$}
 
We briefly indicate some of the features of the construction of this paper that can prevent it
from producing good codes from higher-degree surfaces.  The main issue is that even
if we can control the presence of reducible plane sections, if $X$ is a smooth 
surface of degree $m$ and $H$ is a plane, then provided that $X \cap H$ is 
a smooth curve, it has genus $g = \frac{(m-1)(m-2)}{2}$.  For $q$ fixed, 
the Hasse-Weil-Serre upper bound on $|(X\cap H)(\F_q)|$ used to derive \eqref{d1bd}
grows like a constant times $m^2$ and the minimum distance of the resulting codes often seems 
to decrease too rapidly with $m$ to produce codes better than the best known examples.  

\begin{expr} 
For any $m \ge 5$, consider the surface in $\P^3$ given by 
\[
X_m = V(w^m+xy^{m-1}+yz^{m-1}+zx^{m-1}).
\]
Our motivation for looking at surfaces of this form was that in \cite{s}, Shioda 
proved that these have Picard number one over $\CC$.  On the other hand, for $m=4$, 
this equation defines a $K3$ surface with geometric Picard number 20.
Since these surfaces are defined over $\ZZ$, we can also consider their reductions mod $p$ 
for any prime $p$ and use them to construct codes, although the Picard numbers of the 
reduced surfaces may be larger than 1.  

Even in the case $m = 4$, there are $q$ for which the surface $X_4$ 
contains no lines or conics defined over $\F_q$ and hence has no reducible plane sections.
For instance, with $q = 11$ and $k=1$, we obtain a code $C(X_4,1,\F_{11})$
with parameters $[144,4,120]$.  The minimum weight codewords come from smooth 
plane quartic curves ($g = 3$) with $24$ points rational over $\F_{11}$.  These are optimal
for $g = 3$ according to \cite{mp}.  The 
optimal number of $\F_{11}$-rational points on a curve with $g = 1$ is $18$.  Hence
codes from a $C_{12}$ cubic surface over $\F_{11}$ containing optimal cubic plane sections 
has the same $n = 144$ and $k = 4$, but $d = 144 - 18 = 126$.   

The cases $m \ge 5$ are similar.  For instance, the surface $X_5$ over $\F_9$ contains 
no $\F_9$-rational lines or conics.  The $C(X_5,1,\F_9)$ code has parameters $[91, 4, 71]$ and  
the nonzero words of minimum weight come from irreducible plane quintic curve sections.
Although curves of genus $g = 6$ with as many as $32$ $\F_9$-rational points are known according
to \cite{mp}, none of those are contained in this surface.  But the best known and best possible
$d$ in this cases is $d = 79$ according to \cite{g}.
\end{expr}

\begin{expr} 
In \cite{v}, van Luijk provides examples for the case not covered in Shioda's results, 
giving an explicit family of quartic $K3$ surfaces in $\P^3$ over $\QQ$ with geometric 
Picard number one.  (Although such surfaces were long known to exist, the fact that this
example was only published in 2007 is a reflection of the difficulty of finding
explicit examples.)  
Let
$$
\begin{array}{ccc}
f_1 &=& x^3-x^2y-x^2z+x^2w-xy^2-xyz+2xyw+xz^2+2xzw+y^3+y^2z\\
    & &-y^2w+yz^2+yzw-yw^2+z^2w+zw^2+2w^3\\
f_2 &=& xy^2+xyz-xz^2-yz^2+z^3\\
g_1 &=&z^2+xy +yz \\
g_2 &=& z^2+xy\\
\end{array}
$$
Then for any homogeneous quartic $h$, the surface 
\[
wf_1+2zf_2 - 3g_1g_2 +6h =0 
\]
is a smooth $K3$ surface with Picard number one over $\QQ$.
Once again, if $h$ has coefficients in $\ZZ$, we can look at reductions of such surfaces 
mod $p$ and use them to construct codes. 

We generated 50 such surfaces $X$ randomly over $\F_7$.  Of these, 38 were smooth, and those
surfaces had a wide range of values of $|X(\F_7)|$ -- as small as $38$ and as large as
$69$.  The corresponding $C(X,1,\F_7)$ codes for the surface with the maximal number of points
had parameters $[69,4,52]$.  We mention the following easy observation.  
Smooth plane curves of degree $m = 4$ have genus $3$, and the largest possible number 
of $\F_7$-rational points is $20$ by 
\cite{mp}.  The same is true for any irreducible plane quartic.  
It follows that if any such surface yields a code with $n - d \ge 21$, then 
some plane section must be \emph{reducible}.  The irreducible components of such plane sections cannot
give classes that are integer multiples of the hyperplane class in ${\rm NS}(X)$, so $[H]$ does
not generate that group.  
\end{expr}

\begin{disc}
We think these examples show that the larger sectional genera of quartic and quintic surfaces 
will tend to produce codes with smaller $d$ than codes from cubic surfaces having the same $n$.
Moreover, this tends to be true even when none of the irreducible plane sections of the surface 
are optimal curves of genus $g = \frac{(m-1)(m-2)}{2}$.
On the other hand, some surfaces of higher degree may have many more $\F_q$-rational points than
any cubic surface and that might tend to overcome the effects of the greater sectional genus.  
So the situation is still rather subtle and we cannot say that higher degree surfaces are automatically bad
in the construction of this paper.         
\end{disc}

\section{Restricting $\rho(X)$ and the sectional genus}

From the previous section, we see that to obtain good codes from an algebraic surface, 
in addition to restricting the Picard number, it is probably also necessary to 
restrict the sectional genus $g$ of the surface.  For small values of $g$, this 
puts some severe restrictions on the available classes of surfaces. The following 
results are very well-known; the article \cite{ab} provides a convenient reference
with arguments valid in characteristic $p > 0$.  
We use the notation from that paper where $g(L)$ denotes the arithmetic genus of the curves
in the linear system corresponding to an ample line bundle $L$.  

\begin{thm}[\cite{ab}, Theorem 1.5]
If $X$ is a smooth surface over an algebraically closed field of 
characteristic $p > 0$ and $L$ is an ample line bundle with $g(L) = 0$, then 
(up to isomorphism) $(X,L)$ is one of the following:
\begin{enumerate}
\item $(\P^2, \mathcal{O}_{\P^2}(r))$, $r = 1,2$.
\item $(Q, \mathcal{O}_Q(1))$, where $Q$ is a smooth quadric surface in $\P^3$, or
\item $(F_r, \mathcal{O}_{F_r}(E + sf))$, where $r \ge 1$, $F_r$ is a Hirzebruch surface, $E$ 
is a rational curve on $F_r$ with $E^2 = -r$, $f^2 = 0$, $E\cdot f = 1$, and $s \ge r + 1$.   The 
sections of the given line bundle embed the surface as a rational scroll of degree $\delta = 2s - r$ 
in $\P^{\delta + 1}$.  
\end{enumerate}
\end{thm}

In other words, there are very few examples, and the codes from those surfaces are well-understood from 
the perspective of coding theory as projective Reed-Muller codes, codes from 
quadrics and scrolls, or toric surface codes.

The results of applying this construction to more general rational surfaces have been discussed in 
several works.  In \cite{d}, Davis studied codes obtained from certain very explicit 
rational surfaces with an eye toward proving uniform bounds on the minimum distance.
Let $X$ be the blow up of $\P^2$ at eight or fewer points, and let $E_0$ be the class of the
proper transform of a line, and for $i\ge 1$, let $E_i$ be the exceptional curve over a
blown up point. To obtain the minimum distance, \cite{d} employs a strategy along the lines 
of \cite{ls}: find sections which have maximal possible factorization. The result is 

\begin{thm}[\cite{d}]
Let $T \subset \P^2(\F_q)$ be the ``big torus,'' namely 
$$T = \{(x_0:x_1:x_2) \in \P^2(\F_q) \mid x_0 x_1 x_2 \ne 0\}.$$
If $D = m E_0-\sum m_i E_i$ is a nef divisor and $q \ge \max\{m+2, 2m-\sum m_i\}$, then 
the minimum distance of the $C(T,D,\F_q)$ code satisfies
\[
d \ge (q-1)^2 -m(q-1).
\]
\end{thm}
In \cite{b}, this bound is improved by one and
Couvreur applies many of these ideas to produce some quite good codes in \cite{c}.

The case of sectional genus $g = 1$ seems to be a favorable one for these coding theory applications. 
But again there are actually very few examples as we see from \cite{ab}.

\begin{thm}[\cite{ab}, Theorem 1.8]
If $X$ is a smooth surface over an algebraically closed field of characteristic
$p > 0$ and $L$ is an ample line bundle with $g(L) = 1$, then $(X,L)$ is either a 
elliptic scroll (that is, a ruled surface over an elliptic curve), or a Del Pezzo surface.
\end{thm} 

Scrolls never seem to yield good results in our construction because $\rho(X) \ge 2$ and
there are always hyperplane sections containing several fibers of the ruling.  So
we will concentrate on the Del Pezzo case.   

\subsection{The Del Pezzo quartic in $\P^4$}

In this subsection, we start by considering a code 
obtained from a surface which is a blowup of $\P^2$ at five 
general points. Over a finite field, these surfaces were classified by Rybakov in \cite{r}. 
One reason to consider these surfaces is that the higher embedding dimension
means that the dimension of the $C(X,1,\F_q)$ codes will be larger.   

\begin{expr}
With a random search, we found the following X, a Del Pezzo surface of degree 4 in $\P^4$, 
defined as
\begin{eqnarray*}
{\bf V}(2x^2-2vy-2xy-2y^2+3vz-xz-2yz-2z^2-2xw+2yw+2zw+3w^2,\\
 \qquad -v^2+2vx-x^2-vy-xy-vz+2yz+2z^2-3vw-2xw+2yw-2zw-2w^2).
\end{eqnarray*}
It can be checked easily that 
\begin{itemize}
\item $X$ is smooth
\item $X$ has 57 $\F_7$-rational points
\item $X$ contains no $\F_7$-rational lines or conics.
\end{itemize}
The corresponding $C(X,1,\F_7)$ code has parameters $[57,5,44]$.  

Since there are no $\F_7$-rational lines or conics contained in this surface, 
all the hyperplane sections by $\F_7$-rational planes are irreducible.
As in the discussion of cubic surfaces before, $n - d \le 13$ for any such surface since 
the hyperplane sections of $X$ are elliptic quartic curves in $\P^3$ ($g = 1$).  There are optimal 
$g=1$ curves that appear as hyperplane sections for this surface, so $d = 44$.  By way of contrast, though,
since $X$ is a surface in $\P^4$, we have $k = 5$.  The resulting code has parameters equal to the best 
known code in \cite{g} for this $n, k$ over $\F_7$.  
\end{expr}

\subsection{Codes from other surfaces with sectional genus one}

For a different sort of example related
to the approach taken in \cite{d} and \cite{c},  
take cubics in $\P^2$ through three points defined over $\F_{q^3}$ in 
a general Frobenius orbit ${\mathcal P}_3 = \{P, F(P), F^2(P)\}$.  We assume in particular
that the points in ${\mathcal P}_3$ are \emph{not collinear}. 
If we use the system of cubics in $\P^2$ containing ${\mathcal P}_3$
to map $\P^2$ into $\P^6$, we obtain a certain type of
Del Pezzo surface $X$ of degree 6 defined over $\F_q$.  
(These surfaces are discussed in general in \cite{do}.)
 
\begin{prop}  
This surface $X$ has Picard number $\rho(X) = 2$.  
\end{prop}

\begin{proof}
We start by computing the zeta function of the surface $X$.
 Since $X$ is $\P^2$ blown up in a Frobenius orbit of three points defined
over $\F_{q^3}$, to count points over extensions of $\F_q$, we just look at
the corresponding numbers of points on $\P^2$, then add in the $\F_{q^r}$
rational points on three copies of $\P^1$ defined over $\F_{q^3}$ (other
than the points we blew up).  This gives:
$$
|X(\F_{q^r})| =  \begin{cases} 1 + q^{2r} + q^r & r \equiv 1,2 \bmod 3\\
1 + q^{2r} + 4q^r & r \equiv 0 \bmod 3.\end{cases}
$$
This implies the zeta function of $X$ equals the zeta function
of $\P^2$ over $\F_q$, multiplied by 
$\displaystyle \frac{1}{1 - q^3 t^3}.$  (This is a special case of Lemma 2.7 in \cite{r}.)
Hence because of the form of the zeta function described by the Weil conjectures (proved by Deligne), 
$$Z(X,t) = \frac{1}{(1 - t)P_2(t)(1 - q^2t)},$$
where $P_2(t)$ (the factor corresponding to the $H^2$ cohomology) factors as
$$P_2(t) = (1 - qt)\prod_{j=1}^3 (1 - \alpha_j t),$$  
with $\alpha_j = q, e^{2\pi i/3} q, e^{4\pi i/3} q$.  

By the results of Tate from
\cite{t}, this is enough to imply that $\rho(X) =  2$.  In \cite{t}, Tate conjectured that the Picard number
of any smooth projective surface over $\F_q$ is equal to the multiplicity of $q$ as a reciprocal root of the ``$H^2$-factor''
in the denominator of the zeta function.  Although the equality remains conjectural 
in general, Tate also showed the Picard number is bounded above by the number of reciprocal roots equal
to $q$ in all cases.   The bound is 2 in this case, but we can also see two independent elements of ${\rm NS}(X)$ by 
considering the proper transforms on $X$ of lines in $\P^2$, and of the conics through the 
three points in ${\mathcal P}_3$.   We believe Tate's conjecture was probably known previously
for this type of surface but we were not able to find a convenient reference.  In any case, the bound 
suffices for our purposes.  
\end{proof}

Results from \cite{c} show the $C(X,1,\F_q)$ codes on these surfaces are already quite good.  
For completeness, we indicate the proof from this point of view.

\begin{prop}[\cite{c}, Theorem 3.8 is equivalent to this]
For all prime powers $q$,
$C(X,1,\F_q)$ is a $[q^2 + q + 1, 7, q^2 - q - 1]$ linear
code over $\F_q$. 
\end{prop}
         
\begin{proof}
The sections of $X$ by hyperplanes in $\P^6$ are the images of cubics in the plane
passing through ${\mathcal P}_3$.  
Since ${\mathcal P}_3$ is general, there are no $\F_q$-rational lines containing
all three points in ${\mathcal P}_3$, but there is a 2-parameter family of 
$\F_q$-rational conics containing those points.  Hence the cubics through
${\mathcal P}_3$ contain reducible cubics consisting of an $\F_q$-rational 
conic through ${\mathcal P}_3$ union an arbitrary $\F_q$-rational line.    
Among the codewords of the $C(X,1,{\mathcal P}_3)$ code, then, we will always
find some coming from cubics that factor as an $\F_q$-rational conic containing 
${\mathcal P}_3$ and an $\F_q$-rational line that intersects the conic in two conjugate
points defined over $\F_{q^2}$.  This gives codewords with $2(q + 1) = 2q + 2$ 
zeroes.  We claim that this is the largest possible number of zeroes in any nonzero
codeword, which yields $d = (q^2 + q + 1) - (2q + 2) = q^2 - q - 1$ as in the statement
to be proved.  The claim follows from a case by case analysis of the possible factorizations
of a cubic $C$ containing ${\mathcal P}_3$.  If $C$ is irreducible, the Hasse-Weil-Serre
bound implies 
$$|C(\F_q)| \le 1 + q + \lfloor 2\sqrt{q}\rfloor \le (1 + \sqrt{q})^2.$$
But $(1 + \sqrt{q})^2 < 2 q + 2$
for all $q \ge 2$.  If $C = Q \cup L$ where $Q$ is a conic containing ${\mathcal P}_3$
and $L$ is a line such that $Q \cap L$ consists of distinct $\F_q$-rational points, then
$|(Q\cup L)(\F_q)| = (q + 1) + (q + 1) - 2 = 2q$.   If the line $L$ is tangent to $Q$ then 
$|(Q\cup L)(\F_q)| = (q + 1) + (q + 1) - 1 = 2q + 1$.  So such cubics give codewords containing
$2q$ or $2q + 1$ zeroes.     Finally, there are no reducible cubics with three
$\F_q$-rational line components in this linear system since if such a line contains
$P$, then it would have to contain the whole Frobenius orbit ${\mathcal P}_3$.  But there
are no such lines by our general position hypothesis for ${\mathcal P}_3$.     Hence
the minimum-weight codewords will have $2q + 2$ zeroes as claimed.
\end{proof}

For example, this yields codes with parameters as follows:
\begin{itemize}
\item $q = 7$:  $[57, 7, 41]$ 
\item $q = 8$:  $[73, 7, 55]$ 
\item $q = 9$:  $[91, 7, 71]$
\end{itemize}
All three of these equal the best known values for $d$ for these $n,k$ according to 
\cite{g}.

To conclude our discussion we will consider the $C(X,2,\F_q)$ codes from 
these Del Pezzo surfaces.

\begin{thm}
Let $q > 5$.  Then
$C(X,2,\F_q)$ is a $[q^2 + q + 1, 19, \le q^2 - 3q - 1]$
code over $\F_q$.  Moreover, $d = q^2 - 3q - 1$ for all prime powers $q > 36$.
\end{thm}

\begin{proof}
The dimension of the space of homogeneous polynomials of degree 2 on $\P^6$ is $\binom{6 + 2}{2} = \binom{8}{2}$.
However, it is easy to check that the ideal of $X$ is generated by 9 linearly independent quadrics.
Hence $k = \dim_{\F_q} C(X,2,\F_q) \le \binom{8}{2} - 9 = 19$.  The $\F_q$-rational
points on $X$ are in general enough position to guarantee that Definition~\ref{Fqgen}
is satisfied.  Hence the dimension of the code is exactly $k = 19$.

To show that the minimum distance is bounded above by $q^2 - 3q - 1$, we claim first that there are
codewords with precisely $4q + 2$ zeroes.  Sections of $X$ by quadrics in $\P^6$ correspond
to elements of the linear system of sextic curves in the plane with double points at the three points in 
${\mathcal P}_3$.  Among the elements of that linear system are reducible sextics of the 
form $(C_1 \cup L_1) \cup (C_2 \cup L_2)$ where the $C_i$ are $\F_q$-conics through ${\mathcal P}_3$,
the $L_i$ are $\F_q$-lines, and for each $i$, $C_i \cap L_1$ and $C_i \cap L_2$ both 
consist of two Frobenius-conjugate points
defined over $\F_{q^2}$.  Note that there are $2q + 2$ $\F_q$-rational points on $C_i \cup L_i$ for
each $i$.  However, since $C_1$ and $C_2$ both pass through ${\mathcal P}_3$, their fourth 
point of intersection must be defined over $\F_q$.  Moreover $L_1$ and $L_2$ also meet in an 
$\F_q$-rational point.  Hence the union $(C_1 \cup L_1) \cup (C_2 \cup L_2)$ contains 
$(2q + 2) + (2q + 2) - 2 = 4q + 2$ $\F_q$-rational points.  Evaluating a polynomial corresponding 
to one of these sextics gives a codeword of weight $q^2 + q + 1 - (4q + 2) = q^2 - 3q - 1$.

To conclude the proof, we need to show that  if $q > 36$, no section of $X$ by a quadric in $\P^6$ gives 
a curve with more than $4q + 2$ $\F_q$-rational points.  Sextic curves in the plane with double points
at the points of ${\mathcal P}_3$ have arithmetic genus $g = 7$.  By Corollary
3.3 of \cite{BH}, if $6 \le \sqrt{q}$ (hence the prime power $q$ satisfies $q > 36$), then 
a plane curve of $d = 6$ and genus $g = 7$ has at most
$$\frac{(2\cdot 7 - 2) + (q + 2)\cdot 6}{2} = 3q + 12$$
$\F_q$-rational points.  This result is based on the St\"ohr-Voloch theorem (see \cite{sv})
and the fact that curves of small degree relative to $q$ are automatically Frobenius-classical.  
We see $3q + 12 < 4q + 2$ for all $q > 11$, hence in particular for all $q > 36$.  

If $Y$ is the union of an absolutely irreducible quintic with double points at the points of ${\mathcal P}_3$ and a line, 
then by \cite{ap}, $Y$ has at most $1 + q + 3 \lfloor 2\sqrt{q}\rfloor + 1 + q = 2 + 2q + 6 \lfloor 2\sqrt{q}\rfloor$ $\F_q$-rational points.  
This is $< 4q + 2$ for all $q > 36$.  

If $Y$ is the union of a quartic and a conic, then similarly $Y$ can have at most 
$1 + q + 3 \lfloor 2\sqrt{q}\rfloor + 1 + q = 2 + 2q + 3 \lfloor 2\sqrt{q}\rfloor$ $\F_q$-rational
points and this is also $< 4 q + 2$ for all $q > 9$.  

If $Y$ is the union of two irreducible cubics containing ${\mathcal P}_3$, 
Then $Y$ has at most $2 + 2q + 4\sqrt{q}$ $\F_q$-rational points.  This is $< 4q + 2$ 
for all $q \ge 5$.  If $Y = Z_1 \cup (C \cup L)$
where $Z_1$ is an irreducible cubic over $\F_q$ and $C \cup L$ is a cubic of the form considered in
the second paragraph of this proof, then $Y$ contains at most $1 + q + 2\sqrt{q} + 2q + 2 = 3 + 3q + 2\sqrt{q}$
$\F_q$-rational points.  But $3 + 3q + 2\sqrt{q} < 2 + 4q$ for all $q > 5$.  Finally, any 
$Y = (C_1 \cup L_1) \cup (C_2 \cup L_2)$ has at most $4q + 2$ $\F_q$-rational points as well
since the arrangements considered in the previous paragraph have \emph{smallest number of common
points} defined over $\F_q$ in $(C_1 \cup L_1) \cap (C_2 \cup L_2)$.
\end{proof}

\begin{expr}
\label{q7and9}
By {\tt Magma} calculations we have verified that the equality $d = q^2 - 3q - 1$ also holds for 
$q = 7$ and $q = 9$ using particular choices of the Frobenius orbit 
${\mathcal P}_3 = \{P,F(P),F^2(P)\}$.  This yields a $[57,19,27]$ code over $\F_7$ and 
a $[91,19,53]$ code over $\F_9$.    

The code over $\F_7$ improves on the 
best previously known $d$ by 1.  Markus Grassl observed via another {\tt Magma} calculation 
that this code is equivalent to a cyclic code and the cyclic 
form has been added to the database \cite{g}.   

The code over $\F_9$ improves
the best previously known $d = 51$ by 2.   

On the other hand, for $q = 8$, and again for a particular choice of Frobenius orbit, 
the parameters of the $C(X,2,\F_8)$ code we tested were $[73, 19, 37]$, while $8^2 - 3\cdot 8 - 1 = 39$
and the database \cite{g} contains a $[73,19,39]$ code.  
The minimum-weight codewords for our $[73,19,37]$ code come from absolutely irreducible plane sextics 
with nodes at the three points in ${\mathcal P}_3$ and no other singularities, 
hence curves of geometric genus $7$.  The $36$ $\F_8$-rational points on these curves 
happen to improve on the largest number of $\F_8$-rational points
for curves of genus 7 reported in the database \cite{mp} as of September 2017.  
\end{expr}

In Theorem 3.12 of \cite{c}, Couvreur considers the linear system of all quintic curves in $\P^2$
passing through a Frobenius orbit ${\mathcal P}_3$ and constructs a
$[q^2 + q + 1, 18, q^2 - 3q - 1]$ code.  Our examples improve his by increasing the dimension by $1$.

\section{Concluding remarks} 

Our work raises a number of questions:
\begin{enumerate}
\item Are there other surfaces where Theorem~\ref{d1bdth} and Theorem~\ref{distbds} 
apply to yield really good codes?
\item What happens for varieties of dimension $\ge 3$?
\item Is Conjecture~\ref{C12Conj} true in general?
\end{enumerate}
We hope to return to these in future work.

\vskip 10pt

\noindent{\bf Acknowledgments} This paper began during a visit by the first author 
at the University of Illinois, and we thank NSF for the support that made this possible.
Computations were done using {\tt Magma}, \cite{Mag}.  We would also like to thank
the referees for their careful reading of the original version and their many helpful
comments and suggestions. Finally, we wish to thank Markus Grassl for his
assistance studying the codes described in the Experimental Results~\ref{q7and9}.

\bibliographystyle{siam}

\end{document}